\documentclass[journal,onecolumn,12pt]{IEEEtran}
\usepackage{pslatex}
\usepackage{amsfonts,color,morefloats}
\usepackage{amssymb,amsmath,latexsym,amsthm}

\newcommand{\gf}{{\mathrm{GF}}}

\newtheorem{theorem}{Theorem}
\newtheorem{lemma}[theorem]{Lemma}

\newtheorem{conj}{Conjecture}

\newtheorem{definition}{Definition}

\begin{document}

\title{Power Functions over Finite Fields with Low $c$-Differential Uniformity }
\date{\today}

\author{Haode Yan, Sihem Mesnager, and Zhengchun Zhou
\thanks{H. Yan and Zhengchun Zhou are with the School of Mathematics, Southwest Jiaotong University, Chengdu, 610031, China  (e-mail: hdyan@swjtu.edu.cn, zzc@home.swjtu.edu.cn)}
\thanks{S. Mesnager is with the Department of Mathematics, University of
Paris VIII, 93526 Saint-Denis,  with University
Sorbonne Paris Cit\'e, LAGA, UMR 7539, CNRS,
93430 Villetaneuse,  and
also with the T\'el\'ecom Paris, 91120 Palaiseau, France (e-mail: smesnager@univ-paris8.fr).}
}
\maketitle

\begin{abstract}
Very recently, a new concept called multiplicative differential (and the corresponding $c$-differential uniformity) was introduced by Ellingsen \textit{et al} in
[C-differentials, multiplicative uniformity and (almost) perfect c-nonlinearity, IEEE Trans. Inform. Theory, 2020]
 which is motivated from practical differential cryptanalysis. Unlike classical perfect nonlinear functions, there are
perfect $c$-nonlinear functions even for characteristic two.  The objective of this paper is to study power function $F(x)=x^d$ over finite fields with low $c$-differential uniformity. Some power functions are shown to be perfect $c$-nonlinear or almost perfect $c$-nonlinear. Notably, we completely determine the $c$-differential uniformity of almost perfect nonlinear functions with the well-known Gold exponent. We also give an affirmative solution to a recent conjecture proposed by Bartoli and Timpanella  in 2019 related to an exceptional quasi-planar power function.

%
\end{abstract}
{\bf keywords}:
 Differential uniformity, $c$-differential uniformity, perfect nonlinear function (PN), perfect $c$-nonlinear (P$c$N) function, almost perfect nonlinear function (APN), almost perfect $c$-nonlinear (AP$c$N) function.

\section{Introduction}

Differential cryptanalysis (\cite{BS,BS93}) is one of the most fundamental cryptanalytic approaches targeting symmetric-key primitives. Such a cryptanalysis has attracted a lot of attention since it was proposed to be the first statistical attack for breaking the iterated block ciphers \cite{BS}. The security of cryptographic functions regarding differential attacks was widely studied in the last 30 years. This security is quantified by the so-called \emph{differential uniformity} of the substitution box (S-box) used in the cipher \cite{NK}.  A  very nice survey on the differential uniformity of vectorial Boolean functions can be found in the chapter of Carlet \cite{Cbook} and an interesting article on this topic is \cite{Carlet2018}. In \cite{BCJW}, a new type of differential was proposed. The author's utilized modular multiplication as a primitive operation, which extends the type of differential cryptanalysis. It is necessary to start the theoretical analysis of an (output) multiplicative differential.
Very recently, a new concept called \emph{multiplicative differential }(and the corresponding $c$-differential uniformity) was coined  by Ellingsen \textit{et al} (\cite{EFRST}) which is motivated from practical differential cryptanalysis.

\begin{definition}Let $\gf(p^n)$ denote the finite field with $p^n$ elements, where $p$ is a prime number and $n$ is a positive integer. For a function $F$ from $\gf(p^n)$ to itself, $a,c \in \gf(p^n)$, the (multiplicative) $c$ derivative of $F$ with respect to $a$ is define as
\[_cD_aF(x)=F(x+a)-cF(x), ~\mathrm{for}~\mathrm{all}~x.\]
For $b\in\gf(p^n)$, let $_c\Delta_F(a,b)=\#\{x\in\gf(p^n): F(x+a)-cF(x)=b\}$. We call $_c\Delta_F=\mathrm{max}\{_c\Delta_F(a,b):a,b\in\gf(p^n), \mathrm{and}~ a\neq 0 ~\mathrm{if}~ c=1\}$ the $c$-differential uniformity of F. If $_c\Delta_F=\delta$, then we say $F$ is differentially $(c,\delta)$-uniform.
\end{definition}
If the $c$-differential uniformity of $F$ equals $1$, then $F$ is called a perfect $c$-nonlinear (P$c$N) function. P$c$N functions over odd characteristic finite fields are also called $c$-planar functions. If the $c$-differential uniformity of $F$ is $2$, then $F$ is called an almost perfect $c$-nonlinear (AP$c$N) function. It is easy to see that, for $c=1$ and $a\neq0$, the $c$-differential uniformity becomes the usual differential uniformity, the P$c$N and AP$c$N functions become perfect nonlinear (PN) function and almost perfect nonlinear function (APN) respectively, which play an important role in both theory and application. For even characteristic, APN functions have the lowest differential uniformity. Known APN functions with even characteristic presented in \cite{BD,D1,D2,D3,G,JW,K,N}. For the known results on PN and APN functions over odd characteristic finite fields, the readers are referred to \cite{CM,DMMPW,DO,DY,HRS,HS,L,ZW10,ZW11}.

Because of the strong resistance to differential attacks and the low implementation cost in a hardware environment, power functions $F(x)=x^d$ (i.e., monomials) with low differential uniformity serve as good candidates for the design of S-boxes. Power functions with low differential uniformity may introduce some unsuitable weaknesses within a cipher \cite{BCC,JaKn,CoPi,CaVi}, they also provide better resistance towards differential cryptanalysis. For instance, a differentially $4$-uniform power function, which is extended affine EA-equivalent to the inverse function $x \mapsto x^{2^n-2}$ over $\gf(2^n)$ with even $n$, is employed in the AES (advanced encryption standard).  Two functions $F$ and $F^{\prime}$ from $\gf(p^n)$ to $\gf(p^r)$ are called EA-equivalent if there exist affine automorphisms $L$ from  $\gf(p^n)$ to $\gf(p^n)$ and $L^{\prime}$ from $\gf(p^r)$ to $\gf(p^r)$ and an affine function $L''$ from
$\gf(p^n)$ to $\gf(p^r)$ such that $F'=L^{\prime} \circ F \circ L + L''$.

A nature question one would ask is whether the power functions have good $c$-differential properties. In \cite{EFRST}, the authors studied the $c$-differential uniformity of the well-known inverse function $F(x)=x^{p^n-2}$ over $\gf(p^n)$, for both even and odd prime $p$. It was shown  that $F$ is P$c$N when $c=0$, $F$ is AP$c$N with some conditions of $c$ and differentially $(c,3)$-uniform otherwise. This result illustrates that P$c$N functions can exist for $p=2$. For P$c$N functions $x^{\frac{3^k+1}{2}}$ over $\gf(3^n)$ and $c=-1$, a sufficient and necessary condition is presented in \cite{EFRST}. In \cite{BT}, it is shown that for odd $p$, $n$ and $c=-1$, $x^{\frac{p^2+1}{2}}$ over $\gf(p^n)$ and $x^{p^2-p+1}$ over $\gf(p^3)$ are P$c$N functions. We summarized the known results in Table \ref{table-1} as well as the results obtained in this paper.\\

Planar functions have deep applications in different areas of mathematics. Recently,  Bartoli and  Timpanella  \cite{BT} provided a generalization of planar functions and obtained construction and classification results concerning these new objects. In their paper \cite{BT}, the authors have proposed to solve the following conjecture.

\begin{conj}[\cite{BT}, Conjecture 4.7]\label{conj}Let $p$ be an odd prime and $n$ be an odd integer. For $c=-1$, the power function $x^{\frac{p^n+1}{p+1}}$ is P$c$N over $\gf(p^n)$.
\end{conj}

In this paper, we study the $c$-differential uniformity of power functions. Their $c$-differential uniformity is at most $4$, some of them are P$c$N or AP$c$N. For comparison, we also list the results of this paper in Table \ref{table-1}. The class of P$c$N functions we obtained gives an affirmative solution to Conjecture \ref{conj}. The rest of this paper is organized as follows. In Section II, we study the properties of $c$-differential uniformity of power functions, two useful lemmas and the notation in this paper are also introduced in the section. In Section III, we list the main results together with their proofs. Section IV concludes the paper.
\begin{table}[t]\label{table-1}
\caption{Power functions $F(x)=x^d$ over $\gf(p^n)$ with  low $c$-differential uniformity}
\centering
\begin{tabular}{|c||c|c|c|c|}
\hline
\hline
$p$&$d$ & condition & $_c\Delta_F$ & References \\
[0.5ex]
\hline
any& $2$& $c\neq1$ & 2 &\cite{EFRST}\\
\hline
any &$p^n-2$ &$c=0$ & $1$ &\cite{EFRST}\\
\hline
2 &$2^n-2$ &$c\neq0$, $\mathrm{Tr_n}(c)=\mathrm{Tr_n}(c^{-1})=1$ & $2$ &\cite{EFRST}\\
\hline
2 &$2^n-2$ &$c\neq0$, $\mathrm{Tr_n}(c)=0$ or $\mathrm{Tr_n}(c^{-1})=0$ & $3$ &\cite{EFRST}\\
\hline
odd &$p^n-2$ &$c=4$, $c=4^{-1}$ or $\chi(c^2-4c)=\chi(1-4c)=-1$  & $2$ &\cite{EFRST}\\
\hline
odd &$p^n-2$ &$c\neq0,4,4^{-1}$, $\chi(c^2-4c)$=1 or $\chi(1-4c)=1$ & $3$ &\cite{EFRST}\\
\hline
3& $({3^k+1})/{2}$& $c=-1$, $n/\gcd(k,n)=1$ & 1 &\cite{EFRST}\\
\hline
odd& $({p^2+1})/{2}$& $c=-1$, $n$ odd & 1 &\cite{BT}\\
\hline
odd& $p^2-p+1$& $c=-1$, $n=3$ & 1 &\cite{BT}\\
\hline
2& $2^k+1$& $\gcd(k,n)=1$, $c\neq1$ & 3 &Thm \ref{gold}, This paper\\
\hline
any& $p^k+1$& $1\neq c\in\gf(p^{\gcd(k,n)})$ & $\gcd(k,n)$ &Thm \ref{pk+1}, This paper\\
\hline
odd & $p^k+1$& $\gcd(k,n)=1$, $1\neq c\in\gf(p)$  & 2 &Thm \ref{pk+1}, This paper\\
\hline
odd& $(p^k+1)/2$& $k/\gcd(k,n)$ is even, $c=-1$ & $1$ &Thm \ref{pk+1over2pcn}, This paper\\
\hline
3& $(3^k+1)/2$& $k$ odd, $\gcd(k,n)=1$, $c=-1$ & $2$ &Thm \ref{pk+1over2apcn}, This paper\\
\hline
any& $(2p^n-1)/3$& $p^n\equiv 2 (\mathrm{mod}~3)$, $c\neq1$ & $\leq 3$ &Thm \ref{over3}, This paper\\
\hline
odd& $(p^n+1)/2$& $c\neq\pm1$ & $\leq 4$ &Thm \ref{pn+1over2}, This paper\\
\hline
odd& $(p^n+1)/2$& $c\neq\pm1$, $\chi(\frac{1-c}{1+c})=1$, $p^n\equiv 1 (\mathrm{mod}~4)$ & $\leq 2$ &Thm \ref{pn+1over2}, This paper\\
\hline
odd& $(p^n+3)/2$& $p>3$, $p^n\equiv 3 (\mathrm{mod}~4)$, $c=-1$ & $\leq 3$ &Thm \ref{pn+3over2}, This paper\\
\hline
odd& $(p^n+3)/2$& $p>3$, $p^n\equiv 1 (\mathrm{mod}~4)$, $c=-1$ & $\leq 4$ &Thm \ref{pn+3over2}, This paper\\
\hline
odd& $(p^n-3)/2$& $c=-1$ & $\leq 4$ &Thm \ref{pn-3over2}, This paper\\
\hline
\end{tabular}
\begin{itemize}
    \item $\mathrm{Tr_n}(\cdot)$ denotes the absolute trace mapping from $\gf(2^n)$ to $\gf(2)$.
    \item $\chi(\cdot)$ denotes the quadratic multiplicative character on $\gf(p^n)^*$.
\end{itemize}

\end{table}

\section{Preliminaries}

In this section, we introduce two lemmas which will be used in the sequel.
%
\begin{lemma}\label{power}Let $F(x)=x^d$ be a power function over $\gf(p^n)$. Then
$$
_c\Delta_F=\mathrm{max}\big\{~ \{{_c\Delta_F}(1,b): b\in\gf(p^n) \} \cup \{\gcd(d,p^n-1)\}~\big\}.
$$

\end{lemma}
\begin{proof}For $a=0$ and $c\neq 1$, $_c\Delta_F(0,b)$ is equal to the number of $x\in\gf(p^n)$ such that $x^d=\frac{b}{1-c}$. More precisely,
\begin{eqnarray*}
_c\Delta_F(0,b)=\left\{
\begin{array}{ll}
1, ~~~~~~~~~~~~~~~~~\mathrm{if}~b=0, \\
\gcd(d,p^n-1), ~~\mathrm{if}~\frac{b}{1-c}\in \gf(p^n)^*~\mathrm{is}~\mathrm{a}~d\mathrm{th~power}, \\
0, ~~~~~~~~~~~~~~~~~\mathrm{otherwise}. \\
\end{array} \right.\ \
\end{eqnarray*}
Then $\mathrm{max}\{_c\Delta_F(0,b):b\in\gf(p^n)\}=\gcd(d,p^n-1)$.

For $a\neq0$, two equations $(x+a)^d-cx^d=b$ and $(\frac{x}{a}+1)^d-c(\frac{x}{a})^d=\frac{b}{a^d}$ are equivalent to each other, hence $_c\Delta_F(a,b)=_c\Delta_F(1,\frac{b}{a^d})$. The conclusion then follows from the definition of $_c\Delta_F$.
\end{proof}
To determine the greatest common divisor of integers, the following lemma plays an important role in the rest of this paper.
\begin{lemma}\label{gcd}(Lemma 9,\cite{EFRST}) Let $p,k,n$ be integers greater than or equal to $1$. Then
\begin{eqnarray*}
\gcd(p^k+1,p^n-1)=\left\{
\begin{array}{ll}
\frac{2^{\gcd(2k,n)}-1}{2^{\gcd(k,n)-1}}, ~~~~\mathrm{if}~p=2, \\
2,   ~~~~~~~~~~~~~~\mathrm{if}~\frac{n}{\gcd(n,k)}~\mathrm{is}~\mathrm{odd}, \\
p^{\gcd(k.n)}+1, ~~\mathrm{if}~\frac{n}{\gcd(n,k)}~\mathrm{is}~\mathrm{even}. \\
\end{array} \right.\ \
\end{eqnarray*}
\end{lemma}
We fix some notation and list some facts which will be used in this paper unless otherwise stated.
\begin{itemize}
\item $\gf(p^n)^*$ is the set of nonzero elements in $\gf(p^n)$.
\item $\gf(p^n)^\#=\gf(p^n)\setminus\{0,-1\}$.
\item Let $\chi$ denote the quadratic multiplicative character on $\gf(p^n)^*$.
\item $S_{i,j}:=\{x\in \gf(p^n)^\#: ~\chi(x+1)=i, \chi(x)=j\}$, where $i,j \in \{\pm 1\}$.
\item $S_{1,1}\cup S_{-1,-1}\cup S_{1,-1}\cup S_{-1,1}=\gf(p^n)^\#$.
\item $\Delta(x)=(x+1)^d-cx^d$, $\Delta(0)=1$ and $\Delta(-1)=(-1)^{d+1}c$.
\item $\delta(b)=\#\{x\in\gf(p^n):~\Delta(x)=b\}$.
\end{itemize}
\section{Power Functions with Low c-Differential Uniformity}
In this section, we obtain power functions with low $c$-differential uniformity, some of them come from the power functions with low usual differential uniformity. First, we consider the Gold function over finite field with even characteristic. Let $F(x)=x^d$ be a power function over $\gf(2^n)$, where $d=2^k+1$ and $\gcd(k,n)=1$. It was shown in \cite{G} and \cite{N} that $F(x)$ is an APN function. For the $c$-differential uniformity of Gold function, we have the following theorem.
\begin{theorem}\label{gold}(Gold function) Let $F(x)=x^d$ be a power function over $\gf(2^n)$, where $\gcd(k,n)=1$ and $d=2^k+1$. For $1\neq c\in\gf(2^n)$, $F(x)$ is differentially $(c,3)$-uniform.
\end{theorem}
\begin{proof}For $c\neq1$, we consider $\Delta(x)=b$, i.e.
\begin{equation*}
(x+1)^d+cx^d=b.
\end{equation*}
That is,
\begin{equation}\label{eqngold1}
(1+c)x^{2^k+1}+x^{2^k}+x+1+b=0.
\end{equation}
Let $y=x+\frac{1}{1+c}$, then (\ref{eqngold1}) becomes
\begin{equation*}
y^{2^k+1}+\frac{c+c^{2^k}}{(1+c)^{2^k+1}}y+\frac{bc+b+c}{(1+c)^2}=0.
\end{equation*}
Obviously, there exists $c_0\in\gf(2^n)$ such that $c^{2^k}_0=\frac{c+c^{2^k}}{(1+c)^{2^k+1}}$. Let $z=\frac{y}{c_0}$, then $z$ satisfies
\begin{equation}\label{eqngold2}
z^{2^k+1}+z+v_{b,c}=0,
\end{equation}
where $v_{b,c}=\frac{bc+b+c}{(1+c)^2c^{2^k+1}_0}$. By \cite{HK}, the number of solutions $z\in\gf(2^n)$ of (\ref{eqngold2}) is at most 3. Then $\delta(b)\leq3$ since $x=c_0z+\frac{1}{1+c}$ is bijective. Moreover, $v_{b,c}=\frac{1}{(1+c)c^{2^k+1}_0}b+\frac{c}{(1+c)^2c^{2^k+1}_0}$, when $b$ runs through $\gf(2^n)$, so does $v_{b,c}$. Then there exists some $b$ such that (\ref{eqngold2}) has 3 solutions, i.e., $\delta(b)=3$. This with $\gcd(2^k+1,2^n-1)\leq3$ leads to $F(x)$ is differentially $(c,3)$-uniform.
\end{proof}

In the following, we generalize the Gold function to finite fields with any characteristic. Let $F(x)=x^{p^k+1}$ over $\gf(p^n)$, the $c$-differential uniformity can be determined for some $c$. We obtain AP$c$N functions in this case.
\begin{theorem}\label{pk+1}Let $F(x)=x^d$ be a power function over $\gf(p^n)$, where $d=p^k+1$. For $1\neq c\in\gf(p^{\gcd(k,n)})$, $F(x)$ is differentially $(c,e)$-uniform, where $e=\gcd(d,p^n-1)$. Particularly, if $\gcd(k,n)=1$ and $p$ is odd, $F(x)$ is AP$c$N.
\end{theorem}
\begin{proof}Note that
\begin{align*}
\Delta(x)=&(x+1)^d-cx^d\\
=&(1-c)x^{p^k+1}+x^{p^k}+x+1\\
=&(1-c)(x^{p^k+1}+\frac{1}{1-c}x^{p^k}+\frac{1}{1-c}x)+1\\
=&(1-c)(x+\frac{1}{1-c})^{d}+\frac{c}{c-1}.
\end{align*}
The last identity holds since $\frac{1}{1-c}\in\gf(p^{\gcd(k,n)})$. Then $\Delta(x)$ is a shift of $x^d$. For $b\in \gf(p^n)$, $\delta(b)\leq\gcd(d,p^n-1)=e$ and there exists some $b$ such that the equality holds. By Lemma \ref{power}, $F(x)$ is $(c,e)$-uniform. If $\gcd(k,n)=1$ and $p$ is odd, $e=2$ and consequently $F(x)$ is AP$c$N.
\end{proof}

It is shown in \cite{CM} that power function $F(x)=x^d$ over $\gf(3^n)$ is PN, where $d=\frac{3^k+1}{2}$, $k$ is odd and $\gcd(n,k)=1$. We generalize this class of power functions and obtain P$c$N and AP$c$N functions when $c=-1$. Two theorems are listed as follows.

\begin{theorem}\label{pk+1over2pcn}Let $p$ be an odd prime and $F(x)=x^d$ be a power function over $\gf(p^n)$, where $d=\frac{p^k+1}{2}$. Then
$F(x)$ is PcN for $c=-1$ if and only if $\frac{k}{\gcd(k,n)}$ is even.
\end{theorem}
\begin{proof} If $\frac{k}{\gcd(k,n)}$ is even, then $k$ is even and $d$ is odd. Moreover, $\frac{n}{\gcd(k,n)}$ is odd, then $\gcd(p^k+1,p^n-1)=2$ by Lemma \ref{gcd}, hence $\gcd(d,p^n-1)=1$. First we consider the function $\Delta(x)=(x+1)^d+x^d$ on $\gf(p^n)^\#$. For $x\in\gf(p^n)^\#$, there exist $\alpha,\beta\in\gf(p^{2n})^*$ such that $x+1=\alpha^2$ and $x=\beta^2$. Let $\alpha-\beta=\theta\in\gf(p^{2n})^*$, then $\alpha+\beta=\theta^{-1}$, $\alpha=\frac{1}{2}(\theta+\theta^{-1})$, $\beta=-\frac{1}{2}(\theta-\theta^{-1})$ and $x=\frac{1}{4}(\theta-\theta^{-1})^2$. It can be verified that $x\in\gf(p^n)^*$ if and only if $\theta^{p^n-1}=\pm1$ or $\theta^{p^n+1}=\pm1$. We have
\begin{align*}
\Delta(x)
=&\alpha^{p^k+1}+\beta^{p^k+1}\\
=&\frac{1}{4}(\theta+\theta^{-1})^{p^k+1}+\frac{1}{4}(\theta-\theta^{-1})^{p^k+1}\\
=&\frac{1}{2}(\theta^{p^k+1}+\theta^{-p^k-1}).
\end{align*}
We mention that for fixed $x$, the pair $(\alpha,\beta)$ has $4$ choices and then $x$ corresponds to four $\theta$'s ($\pm\theta,\pm\theta^{-1}$). Although we can choose different $\theta$, $\Delta(x)=\frac{1}{2}(\theta^{p^k+1}+\theta^{-p^k-1})$ always holds, no matter which $\theta$ is chosen. We assume that there exists $x_1\in\gf(3^n)^*$ such that $\Delta(x)=\Delta(x_1)$, i.e., $x$ and $x_1$ respectively correspond to $\theta,\theta_1\in\gf(p^{2n})^*$ , such that
$\frac{1}{2}(\theta^{p^k+1}+\theta^{-p^k-1})=\frac{1}{2}(\theta^{p^k+1}_1+\theta^{-p^k-1}_1)$. It can be obtained that $(\theta\theta_1)^{p^k+1}=1$ or $(\frac{\theta}{\theta_1})^{p^k+1}=1$. Since $\frac{k}{\gcd(k,n)}$ is even, $\frac{2n}{\gcd(2n,k)}$ is odd, therefore $\gcd(p^k+1,p^{2n}-1)=2$ by Lemma \ref{gcd}. Hence $\theta_1=\pm\theta$ or $\pm\theta^{-1}$. Then $x_1=\frac{1}{4}(\theta_1-\theta_1^{-1})^2=\frac{1}{4}(\theta-\theta^{-1})^2=x$. This means that the mapping $\Delta(x)$ is bijective on $\gf(p^n)^\#$.

Now we consider the function $\Delta(x)$ on $x=0$ and $x=1$. It is clear that $\Delta(0)=1$ and $\Delta(-1)=-1$ since $d$ is odd. For $b=1$, if there exists $x\in\gf(p^n)^\#$ such that $\Delta(x)=1$, then we can find $\theta\in\gf(p^{2n})^*$ related to $x$, which satisfies $\frac{1}{2}(\theta^{p^k+1}+\theta^{-p^k-1})=1$. Then $\theta^{p^k+1}=1$, which implies that $\theta=\pm1$, $x=\frac{1}{4}(\theta-\theta^{-1})^2=0$, a contradiction. For $b=-1$, if there exists $x\in\gf(p^n)^\#$ such that $\Delta(x)=-1$, then we can find $\theta\in\gf(p^{2n})^*$ related to $x$, which satisfies $\frac{1}{2}(\theta^{p^k+1}+\theta^{-p^k-1})=-1$. Then $\theta^{p^k+1}=-1$, which implies that $\theta^2=-1$, $x=\frac{1}{4}(\theta-\theta^{-1})^2=-1$, a contradiction. By discussions as above, $\gcd(d,p^n-1)=1$ and $\Delta(x)$ is bijective on $\gf(p^n)$, then $F(x)$ is P$c$N.

If $F(x)$ is P$c$N, then $\gcd(d,p^n-1)=1$. By Lemma \ref{gcd}, $\frac{n}{\gcd(n,k)}$ should be odd and $p^k \equiv 1 (\mathrm{mod}~4)$. Moreover, $\Delta(x)=1$ has unique solution $x=0$ in $\gf(p^n)$.
If $\frac{2n}{\gcd(2n,k)}$ is even, then $\gcd(p^k+1,p^{2n}-1)=p^{\gcd(k,n)}+1$ by Lemma \ref{gcd}. Then we can obtain $\theta^{p^{\gcd(k,n)}+1}=1$ from $\theta^{p^k+1}=1$. There exists $\theta\in\gf(p^{2n})^*$ which satisfies $\theta\neq\pm1$ and $x=\frac{1}{4}(\theta-\theta^{-1})^2\in\gf(p^n)^*$, this means that $\Delta(x)=1$ has more than one solutions in $\gf(p^n)$, a contradiction. Then $\frac{2n}{\gcd(2n,k)}$ is odd, this leads to $\frac{k}{\gcd(k,n)}$ is even, which contains the condition $p^k \equiv 1 (\mathrm{mod}~4)$.

\end{proof}

{\bf Remark 1}. In \cite{EFRST}, the authors proved that when $p=3$, $d=\frac{3^k+1}{2}$ and $c-1$, $F(x)=x^d$ is P$c$N over $\gf(p^n)$ if and only if $\frac{n}{\gcd(n,k)}$ is odd. However, it seems that the condition is not strong enough. If $k$ is odd, $d=\frac{3^k+1}{2}$ is even, then $\gcd(d,3^n-1)\geq2$. By Lemma \ref{power}, the $c$-differential uniformity of $F(x)$ is at least $2$, which is not a P$c$N function.

Now we consider Conjecture \ref{conj}. For odd prime $p$ and odd $n$, $d=\frac{p^n+1}{p+1}$ is odd, and then $\gcd(d,p^n-1)=1$. Note that multiplicative inverse in $\gf(p^n)^*$ of $d$ is $\frac{p(p^{n-1}+1)}{2}$, power functions $x^d$ and $x^{\frac{p^{n-1}+1}{2}}$ are equivalent to each other. By Theorem \ref{pk+1over2pcn}, the power function $x^{\frac{p^{n-1}+1}{2}}$ is P$c$N, then $x^d$ is P$c$N. This gives an affirmative solution to Conjecture \ref{conj}.
\begin{theorem}\label{pk+1over2apcn}Let $F(x)=x^d$ be a power function over $\gf(3^n)$, where $d=\frac{3^k+1}{2}$. Then
$F(x)$ is APcN for $c=-1$ if $k$ is odd and $\gcd(k,n)=1$.
\end{theorem}
\begin{proof}
Since $\gcd(k,n)=1$, by Lemma \ref{gcd}, $\gcd(3^k+1,3^n-1)=2$ when $n$ is odd and $\gcd(3^k+1,3^n-1)=4$ when $n$ is even. Note that $k$ is odd, $4|3^k+1$ and $8\nmid 3^k+1$, then we have $\gcd(d,3^n-1)=2$ for all $n$. First we consider the function $\Delta(x)=(x+1)^d+x^d$ on $\gf(3^n)^\#$. For $x\in\gf(3^n)^\#$, there exist $\alpha,\beta\in\gf(3^{2n})^*$ such that $x+1=\alpha^2$ and $x=\beta^2$. Let $\alpha-\beta=\theta\in\gf(3^{2n})^*$, then $\alpha+\beta=\theta^{-1}$, $\alpha=-(\theta+\theta^{-1})$, $\beta=\theta-\theta^{-1}$ and $x=(\theta-\theta^{-1})^2$. It can be verified that $x\in\gf(3^n)^*$ if and only if $\theta^{3^n-1}=\pm1$ or $\theta^{3^n+1}=\pm1$. We have $\Delta(x)=-(\theta^{3^k+1}+\theta^{-3^k-1})$ always holds no matter the choices of $\alpha$ and $\beta$.

We assume that there exists $x_1\in\gf(3^n)^*$ such that $\Delta(x)=\Delta(x_1)$, i.e., $x$ and $x_1$ respectively correspond to $\theta$ and $\theta_1\in\gf(3^{2n})^*$ , such that
$-(\theta^{3^k+1}+\theta^{-3^k-1})=-(\theta^{3^k+1}_1+\theta^{-3^k-1}_1)$. It can be obtained that $(\theta\theta_1)^{3^k+1}=1$ or $(\frac{\theta}{\theta_1})^{3^k+1}=1$. Since $k$ is odd, $\frac{2n}{\gcd(k,2n)}$ is even, therefore $\gcd(3^k+1,3^{2n}-1)=3^{\gcd(k,2n)}+1=4$ by Lemma \ref{gcd}. Hence $\theta_1=\delta^i\theta$ or $\delta^i\theta^{-1}$, $i=0,1,2,3$, where $\delta\in\gf(3^{2n})^*$ is a $4$th root of unity. We can verify that $x_1\in\gf(3^n)$ for such $\theta_1$'s. It can be seen that for $\theta_1=\pm\theta,\pm\theta^{-1}$, $x_1=(\theta-\theta^{-1})^2=x$ and for $\theta_1=\pm\delta\theta,\pm\delta\theta^{-1}$, $x_1=-(\theta+\theta^{-1})^2$. If $-(\theta+\theta^{-1})^2= (\theta-\theta^{-1})^2$, i.e., $\theta^2+\theta^{-2}=0$, we have $x=(\theta-\theta^{-1})^2=1$ and $\Delta(1)=-1$. It means that the mapping $\Delta(x)$ is $2-\mathrm{to}-1$ on $\gf(3^n)\setminus\gf(3)$, and $\Delta(x)=-1$ has unique solution $x=1$ in $\gf(3^n)^\#$.

It is clear that $\Delta(0)=\Delta(-1)=1$ since $d$ is even. For $b=1$, if there exists $x\in\gf(3^n)^\#$ such that $\Delta(x)=1$, then we can find $\theta\in\gf(3^{2n})^*$ related to $x$, which satisfies $-(\theta^{3^k+1}+\theta^{-3^k-1})=1$. Then $\theta^{3^k+1}=1$, which implies that $\theta=\pm1,\pm\delta$, where $\delta\in\gf(3^{2n})^*$ we defined before. Hence $x=(\theta-\theta^{-1})^2=0$ or $-1$, a contradiction. Then we proved that there is no solution in $x\in\gf(3^n)\setminus\{0,-1\}$ such that $\Delta(x)=1$. By discussions as above, $F(x)$ is AP$c$N.

\end{proof}

In what follows, we obtain power functions with low $c$-differential uniformity via power functions with low usual differential uniformity. In \cite{HRS}, it was proved that if $p^n\equiv 2 (\mathrm{mod}~3)$, the power function $x^d$ is an APN function over $\gf(p^n)$, where $d=\frac{2p^n-1}{3}$. We study the $c$-differential uniformity in the following.
\begin{theorem}\label{over3}Let $F(x)$ be a power function over $\gf(p^n)$, where $d=\frac{2p^n-1}{3}$ and $p^n \equiv 2 (\mathrm{mod}~3)$. For $c\neq1$, $_c\Delta_F\leq 3$.
\end{theorem}
\begin{proof}We know that $\gcd(d,p^n-1)=1$ since $3d-2(p^n-1)=1$. For any $b\in\gf(p^n)$, consider the equation $\Delta(x)=(x+1)^d-cx^d=b$. If $x\in\gf(p^n)$ is a solution of $\Delta(x)=b$, let $x+1=\alpha^3$ and $x=\beta^3$. Such $\alpha,\beta\in\gf(p^n)$ exist uniquely because $\gcd(3,p^n-1)=1$. Then $\alpha$ and $\beta$ satisfy $\alpha^3-\beta^3=1$ and $b=\alpha^{3d}-c\beta^{3d}=\alpha-c\beta$. Note that $x$ is uniquely determined by $\beta$ and $\beta$ satisfies a cubic equation $(b+c\beta)^3-\beta^3=1$, which has at most $3$ solutions in $\gf(p^n)$. This implies $_c\Delta_F\leq 3$.
\end{proof}


Although the power function $x^{\frac{p^n+1}{2}}$ over $\gf(p^n)$ has high usual differential uniformity, it has low $c$-differential uniformity for all $\pm1\neq c\in\gf(p^n)$.
\begin{theorem}\label{pn+1over2}Let $F(x)=x^d$ be a power function over $\gf(p^n)$, where $d=\frac{p^n+1}{2}$ and $p$ is an odd prime. For $\pm1\neq c\in\gf(p^n)$, $_c\Delta_F \leq 4$. Moreover, if $p^n \equiv 1 (\mathrm{mod}~4)$ and $c$ satisfies $\chi(\frac{1-c}{1+c})=1$, $_c\Delta_F \leq 2$.
\end{theorem}

\begin{proof}For any $b\in\gf(p^n)$, if $x\in\gf(p^n)^\#$ is a solution of $\Delta(x)=(x+1)^d-cx^d=b$, then $x$ satisfies
\begin{equation}\label{eqnpn+1over21}
\chi(x+1)(x+1)-c\chi(x)x=b.
\end{equation}
We distinguish 4 cases.

Case I. $x\in S_{1,1}$, i.e., $\chi(x+1)=\chi(x)=1$. Then (\ref{eqnpn+1over21}) becomes $x+1-cx=b$, i.e. $x=\frac{-b+1}{c-1}$. That means (\ref{eqnpn+1over21}) has at most one solution in $S_{1,1}$.

Case II. $x\in S_{-1,-1}$, i.e., $\chi(x+1)=\chi(x)=-1$. Then (\ref{eqnpn+1over21}) becomes $-x-1+cx=b$, i.e. $x=\frac{b+1}{c-1}$. That means (\ref{eqnpn+1over21}) has at most one solution in $S_{-1,-1}$.

Case III. $x\in S_{1,-1}$, i.e., $\chi(x+1)=1,\chi(x)=-1$. Then (\ref{eqnpn+1over21}) becomes $x+1+cx=b$, i.e. $x=\frac{b-1}{c+1}$. That means (\ref{eqnpn+1over21}) has at most one solution in $S_{1,-1}$.

Case IV. $x\in S_{-1,1}$, i.e., $\chi(x+1)=-1,\chi(x)=1$. Then (\ref{eqnpn+1over21}) becomes $-x-1-cx=b$, i.e. $x=\frac{-b-1}{c+1}$. That means (\ref{eqnpn+1over21}) has at most one solution in $S_{1,-1}$.

We have $\delta(b)\leq4$ for $b\neq1,\pm c$ since $\Delta(0)=1$ and $\Delta(-1)=c$ or $-c$. Note that $\Delta(x)=1$ and $\Delta(x)=c$ have no solution in $S_{1,1}$, and $\Delta(x)=-c$ has no solutions in $S_{1,-1}$, then $\delta(0),\delta(\pm c)\leq 4$. Then $_c\Delta_F \leq 4$ follows by Lemma \ref{power} and $\gcd(\frac{p^n+1}{2},p^n-1)\leq2$.

Now we assume that $p^n \equiv 1 (\mathrm{mod}~4)$ and $\chi(\frac{1-c}{1+c})=1$, then $\chi(-1)=1$. For fixed $b\neq1,\pm c$, If (\ref{eqnpn+1over21}) has solutions in $S_{1,1}$, then $\chi(\frac{-b+c}{c-1})=1$. If (\ref{eqnpn+1over21}) has solutions in $S_{-1,1}$, then $\chi(\frac{-b+c}{c-1})=-1$. We conclude that (\ref{eqnpn+1over21}) cannot have solution in $S_{1,1}$ and $S_{-1,1}$ simultaneously. Similarly, we can prove that (\ref{eqnpn+1over21}) cannot have solution in $S_{-1,-1}$ and $S_{1,-1}$ simultaneously. That means $\delta(b)\leq 2$ for $b\neq1,\pm c$. It can be verified that $\Delta(x)=c$ (respectively, $\Delta(x)=-c$) has no solution in $S_{1,1}$ and $S_{-1,1}$ simultaneously (respectively, in $S_{-1,-1}$ and $S_{1,-1}$), then $\delta(c),\delta(-c)\leq2$. We can prove that (\ref{eqnpn+1over21}) cannot have solution in both $S_{1,1}$ and $S_{1,-1}$. Then $\delta(1)\leq2$ since $\chi(-1)=1$, $\Delta(x)=1$ has no solution in $S_{-1,-1}$ and $S_{1,-1}$ simultaneously. Hence $_c\Delta_F \leq 2$.
\end{proof}

When we study the $c$-differential properties, $c=-1$ is a very special case. It can be seen that when $c=-1$, the $c$-differential equation becomes $\Delta(x)=(x+a)^d+x^d$. Sometimes the power function has low $c$-differential uniformity when $c=-1$.

In \cite{HS}, the authors studied the differential uniformity of power function $x^{\frac{p^n+3}{2}}$ and $x^{\frac{p^n-3}{2}}$ over $\gf(p^n)$. We consider their $c$-differential uniformity when $c=-1$. For $p=3$, power function $x^\frac{3^n+3}{2}$ is equivalent to $x^\frac{3^{n-1}+1}{2}$, which was studied in Theorems \ref{pk+1over2pcn} and \ref{pk+1over2apcn}. In the following, we discuss the $c$-differential uniformity of $x^{\frac{p^n+3}{2}}$ for $c=-1$ and $p>3$.
\begin{theorem}\label{pn+3over2}Let $F(x)=x^d$ be a power function over $\gf(p^n)$, where $p>3$ is an odd prime and $d=\frac{p^n+3}{2}$. For $c=-1$, $_c\Delta_F \leq 4$ if $p^n\equiv 1 (\mathrm{mod}~4)$ and $_c\Delta_F \leq 3$ if $p^n\equiv 3 (\mathrm{mod}~4)$.
\end{theorem}
\begin{proof}It is easy to see that $\gcd(d,p^n-1)=1$ when $p^n\equiv 3 (\mathrm{mod}~4)$ and $\gcd(d,p^n-1)=2$ or $4$ when $p^n\equiv 1 (\mathrm{mod}~4)$. For $b\in\gf(p^n)$, if $x\in\gf(p^n)^\#$ is a solution of $\Delta(x)=b$, then $x$ satisfies
\begin{equation}\label{eqnpn+3over21}
\chi(x+1)(x+1)^2+\chi(x)x^2=b.
\end{equation}
We distinguish the following four cases.

Case I. $x\in S_{1,1}$, i.e., $\chi(x+1)=\chi(x)=1$. Then (\ref{eqnpn+3over21}) becomes $x(x+1)=\frac{b-1}{2}$, which has at most two solutions.

Case II. $x\in S_{-1,-1}$, i.e., $\chi(x+1)=\chi(x)=-1$. Then (\ref{eqnpn+3over21}) becomes $x(x+1)=\frac{-b-1}{2}$, which has at most two solutions.

Case III. $x\in S_{1,-1}$, i.e., $\chi(x+1)=1, \chi(x)=-1$. Then we obtain $x=\frac{b-1}{2}$ from (\ref{eqnpn+3over21}).

Case IV. $x\in S_{-1,1}$, i.e., $\chi(x+1)=-1, \chi(x)=1$. Then we obtain $x=\frac{-b-1}{2}$ from (\ref{eqnpn+3over21}).

First we assert that $\Delta(x)=b$ cannot have solutions in $S_{1,1}$ and $S_{1,-1}$ simultaneously for fixed $b\in\gf(p^n)$. Suppose on the contrary, then $\chi(\frac{b-1}{2})=-1$ since it is a solution in $S_{1,-1}$, and $\chi(\frac{b-1}{2})=1$ since (\ref{eqnpn+3over21}) has solutions in $S_{1,1}$, which is a contradiction.

If $p^n\equiv 1 (\mathrm{mod}~4)$, then $\chi(-1)=1$. Then $\Delta(x)=b$ cannot have solutions in $S_{1,1}$ and $S_{-1,1}$ simultaneously. Otherwise, we obtain $\chi(\frac{-b+1}{2})=-1$ from Case IV and $\chi(\frac{b-1}{2})=1$ from Case I, which is a contradiction.  That means for any $b$, the solutions of (\ref{eqnpn+3over21}) in $\gf(p^n)^\#$ is at most $4$. It is easy to see that $\Delta(0)=\Delta(-1)=1$ since $d$ is even. For $b=1$, it can be verified that $\Delta(x)=1$ has no solution in $S_{1,1},S_{1,-1}$ and $S_{-1,1}$, then $\delta(1)\leq4$. This with $\gcd(d,p^n-1)=2$ or $4$ implies that $_c\Delta_F\leq4$.

If $p^n\equiv 3 (\mathrm{mod}~4)$, then $\chi(-1)=-1$. If $x_1\in S_{1,1}$ is a solution of $x(x+1)=\frac{b-1}{2}$, then the other solution is $-x_1-1$. Note that $x_1$ and $-x_1-1$ cannot in $S_{1,1}$ simultaneously, so (\ref{eqnpn+3over21}) has at most $1$ solution in $S_{1,1}$. Similarly, (\ref{eqnpn+3over21}) has at most $1$ solution in $S_{-1,-1}$. Then $\Delta(x)=b$ has at most $3$ solutions in $\gf(p^n)^\#$ since $\Delta(x)=b$ cannot have solutions in both $S_{1,1}$ and $S_{1,-1}$. It is clear that $\Delta(0)=1$ and $\Delta(-1)=-1$ since $d$ is odd. For $b=1$ and $b=-1$, it can be verified that $\Delta(x)=1$ and $\Delta(x)=-1$ have no solution in $\gf(p^n)^\#$. We conclude that $_c\Delta_F\leq3$.
\end{proof}

We also discuss the $c$-differential uniformity of $x^{\frac{p^n-3}{2}}$  for $c=-1$ as follows.

\begin{theorem}\label{pn-3over2}Let $F(x)=x^d$ be a power function over $\gf(p^n)$, where $p$ is an odd prime and $d=\frac{p^n-3}{2}$. For $c=-1$, $_c\Delta_F \leq 4$.
\end{theorem}
\begin{proof}It is easy to see that $\gcd(d,p^n-1)=1$ when $p^n\equiv 1 (\mathrm{mod}~4)$ and $\gcd(d,p^n-1)=2$ when $p^n\equiv 3 (\mathrm{mod}~4)$. For $b\in\gf(p^n)$, we consider the equation
\begin{equation}\label{eqnpn-3over21}
\Delta(x)=(x+1)^d+x^d=b.
\end{equation}
If $b=0$, (\ref{eqnpn-3over21}) has unique solution $x=-\frac{1}{2}$ when $p^n\equiv 1 (\mathrm{mod}~4)$ and has no solution when $p^n\equiv 3 (\mathrm{mod}~4)$. Now we assume $b\neq0$. If $x\in\gf(p^n)^\#$ is a solution of $\Delta(x)=b$ for fixed $b\in\gf(p^n)^*$, then $x$ satisfies
\begin{equation}\label{eqnpn-3over22}
\chi(x+1)(x+1)^{-1}+\chi(x)x^{-1}=b.
\end{equation}
We distinguish the following four cases.

Case I. $x\in S_{1,1}$, i.e., $\chi(x+1)=\chi(x)=1$. Then (\ref{eqnpn-3over22}) becomes $x^2+(1-\frac{2}{b})x-\frac{1}{b}=0$, which has at most two solutions.

Case II. $x\in S_{-1,-1}$, i.e., $\chi(x+1)=\chi(x)=-1$. Then (\ref{eqnpn-3over22}) becomes $x^2+(1+\frac{2}{b})x+\frac{1}{b}=0$, which has at most two solutions.

Case III. $x\in S_{1,-1}$, i.e., $\chi(x+1)=1, \chi(x)=-1$. Then (\ref{eqnpn-3over22}) becomes $x(x+1)=-\frac{1}{b}$,  which has at most two solutions.

Case IV. $x\in S_{-1,1}$, i.e., $\chi(x+1)=-1, \chi(x)=1$. Then (\ref{eqnpn-3over22}) becomes $x(x+1)=\frac{1}{b}$,  which has at most two solutions.

If $p^n\equiv 1 (\mathrm{mod}~4)$, then $\chi(-1)=1$. If $b$ is a square element, then $\chi(-\frac{1}{b})=\chi(\frac{1}{b})=1$. In both Cases III and IV, if $x$ is a solution, then $\chi(x(x+1))=-1$. So (\ref{eqnpn-3over21}) has no solution in $S_{1,-1}$ and $S_{-1,1}$, hence (\ref{eqnpn-3over21}) has at most $4$ solutions in $\gf(p^n)^\#$. If $b$ is a nonsquare element, consider the solutions in each case. In Case I, the product of two solutions of equation $x^2+(1-\frac{2}{b})x-\frac{1}{b}=0$ is $-\frac{1}{b}$, which is a nonsquare element, this means (\ref{eqnpn-3over21}) has at most one solution in $S_{1,1}$. Similarly, we can prove that (\ref{eqnpn-3over21}) has at most 1 solution in $S_{-1,-1}$. Now we consider Case III.
Let $x_3$ and $-x_3-1$ be the two solutions of the quadratic equation $x(x+1)=-\frac{1}{b}$. It is easy to check that $x_3\in S_{1,-1}$ if and only if  $-x_3-1\in S_{-1,1}$. This implies that (\ref{eqnpn-3over21}) has at most $1$ solution in $S_{1,-1}$. Similarly, we can prove that (\ref{eqnpn-3over21}) has at most $1$ solution in $S_{-1,1}$. Now we proved the (\ref{eqnpn-3over21}) has at most 4 solutions in $\gf(p^n)^\#$ for $b\in\gf(p^n)^*$.

It is easy to see that $\Delta(0)=1$ and $\Delta(-1)=-1$. When $b=1$ is a square element, it was proved that $\Delta(x)=1$ has no solutions in Cases III and IV. In Case I, the quadratic equation is $x^2-x-1=0$, $x=\frac{1\pm\sqrt{5}}{2}$. If they are solutions, $\chi(\frac{1\pm\sqrt{5}}{2})=\chi(\frac{3\pm\sqrt{5}}{2})=1$. In Case II, the quadratic equation is $x^2+3x+1=0$, $x=\frac{-3\pm\sqrt{5}}{2}$. If they are solutions, $\chi(\frac{-3\pm\sqrt{5}}{2})=\chi(\frac{-1\pm\sqrt{5}}{2})=-1$, i.e., $\chi(\frac{3\pm\sqrt{5}}{2})=\chi(\frac{1\pm\sqrt{5}}{2})=-1$ since $\chi(-1)=1$. That means (\ref{eqnpn-3over21}) has at most $2$ solutions in $S_{1,1}$ and $S_{-1,-1}$. Then $\delta(1)\leq3$. We can prove $\delta(-1)\leq3$ in a similar way. By discussions as above, we conclude that $_c\Delta_F \leq 4$ if $p^n\equiv 1 (\mathrm{mod}~4)$.

If $p^n\equiv 3 (\mathrm{mod}~4)$, then $\chi(-1)=-1$. Now $d$ is an even number, if $x$ is a solution of $\Delta(x)=(x+1)^d+x^d=b$ for some $b$, so is $-x-1$. This means that the solution number of $\Delta(x)=b$ in $S_{1,1}$ and $S_{-1,-1}$ are the same. If $b$ is a square element, (\ref{eqnpn-3over21}) has at most $1$ solution in $S_{1,1}$ since $\chi(-\frac{1}{b})=-1$, then (\ref{eqnpn-3over21}) has at most $1$ solution in $S_{-1,-1}$. There is no solution in $S_{-1,1}$ since the solution satisfies $\chi(x(x+1))=-1$. Then (\ref{eqnpn-3over21}) has at most $4$ solutions in $\gf(p^n)^\#$ when $b$ is a square element. If $b$ is a nonsquare element, we can similarly prove that (\ref{eqnpn-3over21}) has at most $1$ solution in $S_{1,1}$, at most $1$ solution in $S_{-1,-1}$ and no solution in $S_{1,-1}$. Then we proved that (\ref{eqnpn-3over21}) has at most 4 solutions in $\gf(p^n)^\#$ for $b\in\gf(p^n)^*$.

It is easy to see that $\Delta(0)=\Delta(-1)=1$. Now we focus on $b=1$, which is a square element. It was proved that $\Delta(x)=1$ has at most $1$ solution in $S_{1,1}$, at most $1$ solution in $S_{-1,-1}$, and no solution in Case $S_{-1,1}$. If $\widetilde{x}\in S_{1,-1}$ is a solution of  (\ref{eqnpn-3over21}), then $\chi(\widetilde{x}+1)=1$, $\chi(\widetilde{x})=-1$ and $\widetilde{x}(\widetilde{x}+1)=-1$. Then $\widetilde{x}$ satisfies $(\widetilde{x}+1)^2=\widetilde{x}$, the left-hand side is a square while the right-hand side is a nonsquare, which is a contradiction. Then $\Delta(x)=1$ has no solution in $S_{1,-1}$. That means $\delta(1)\leq4$. By discussions as above, we conclude that $_c\Delta_F \leq 4$ if $p^n\equiv 3 (\mathrm{mod}~4)$, which completes the proof.

\end{proof}

\section{Concluding remarks}
In 2020, Ellingsen \textit{et al.} \cite{EFRST} have defined a new (output) multiplicative differential, and the corresponding $c$-differential uniformity.  Using this new concept, even for characteristic $2$, there are perfect c-nonlinear (PcN) functions. The  modification on the classical notion of differential uniformity was motivated by the use of modular multiplications in some symmetric cryptographic schemes such as the well-known IDEA cipher and other recent symmetric primitives.

In the current paper, we pushed further the successful attempt initiated in \cite{EFRST} by studying the $c$-differential uniformity of power functions over finite fields (which represent an important class of functions due to their low implementation cost in a hardware environment). We derived several classes of power functions with low $c$-differential uniformity. Some of them are PcN or almost PcN. We have also provided a proof of a recent conjecture proposed by Bartoli and Timpanella related to an exceptional quasi-planar power function confirming its validity.
It would be possible and interesting to find more functions over finite fields with low $c$-differential uniformity.


\end{document}